\documentclass[reqno]{amsart} \usepackage{amscd}
\usepackage{epsf}
\newtheorem{theorem}{Theorem}[section]
\newtheorem{proposition}[theorem]{Proposition}

\theoremstyle{remark} 

\begin{document}

\title[Berezin-type quantization on even-dimensional  compact manifolds]{ Berezin-type quantization on even-dimensional  compact manifolds}

\author{ Rukmini Dey* and Kohinoor Ghosh**}
\address{Rukmini Dey,  email:rukmini@icts.res.in}

\address{Kohinoor Ghosh, email: kohinoorghosh@gmail.com}

\maketitle

\begin{center}
\text{I.C.T.S.-TI.F.R., Bangalore, India.*}

\text{I.C.T.S.-TI.F.R., Bangalore, India.**}
\end{center}

\vskip .5in

\vskip .5mm


\begin{abstract}
In this article  we show that a Berezin-type quantization can be achieved on a compact even dimensional manifold $M^{2d}$ by removing a skeleton $M_0$ of lower dimension such that what remains is diffeomorphic to $R^{2d}$ (cell decomposition) which we identify with $C^d$ and embed in $ CP^d$.   A local  Poisson structure and Berezin-type quantization  are induced from $ CP^d$.  Thus we have a Hilbert space with a reproducing kernel.  The symbols of bounded linear operators on the Hilbert space  have a star product which satisfies  the correspondence principle outside a set of measure zero. This construction depends on the diffeomorphism.  One needs to keep track of the global holonomy and hence the  cell decomposition of the manifold.  As an example, we illustrate this type of quanitzation of the torus. We exhibit Berezin-Toeplitz quantization of a complex manifold in the same spirit as above. 
\end{abstract}

\section{Introduction}
Berezin quantization \cite{Be} is a method of defining a star product on the symbol  of operators acting on a Hilbert space (with a reproducing kernel) on a K$\ddot{\rm{a}}$hler manifold under certain conditions such that the star product  satisfies the correspondence principle. The literature on subsequent work after \cite{Be} on Berezin quantization is vast.  We mention that in \cite{E} the conditions have been relaxed considerably. 
Another direction this field has expanded is Berezin-Toeplitz quantization, see for instance \cite{BMS},  \cite{Ko}.

Some quantum systems donot come from quantizing classical systems (which are expected to have a symplectic structure) but there is a semi-classical limit of the quantum system.  For instance there is a semi-classical limit of the quantum system of  spin, see for instance,  \cite{Rad} ($S \rightarrow \infty $ in Radcliffe's notation).   We wish to include systems which donot have symplectic structure (or group action) and study if they are a semi-classical limit of some quantum system as $\hbar$ goes to zero.  This is the motivation of considering manifolds which have no symplectic or Poisson structure.  We induce local Poisson structure  on the manifold  by embedding parts  of it (i.e.  removing sets of measure zero)  into $CP^n$ or $C^n$ (depending on whether we expect a finite dimensional or an infinite dimensional Hilbert space) and induce the Berezin quantization from one of these two spaces. 

The other motivation of the work is that   sometimes the Hilbert space of the problem turns out to be different from what the actual manifold of parameter space should prescribe.  The Hilbert space could be just obtained from geometric quantization of $C^n$, or $ CP^n$,  whereas the parameter space  is not $C^n$ or $CP^n$. Roughly speaking in these two cases (namely $C^n$ or $CP^n$),  the Hilbert space consists of polynomials. For some situations  this could be a at least a  good approximation,  for example the Quantum Hall Effect (where polynomials suffice for lowest Landau levels \cite{To}). The global holonomy needs to be calculated,  which we explain in the example of the torus. 

In this article we show that a Berezin quantization can be achieved on a compact even dimensional manifold $M^{2d}$ by removing a skeleton of lower dimension such that what remains is diffeomorphic to $R^{2d}$ which we identify with $C^d$ and embed in $CP^d$.  We get an induced Berezin quantization from $CP^d$.  In other words,  we obtain a Hilbert space with a reproducing kernel and a star product on the symbol of bounded linear operators on the Hilbert space which satisfy the correspondence principle. 
The Berezin quantization depends on the diffeomorphism  of $M \setminus M_0$  to $R^{2d}$ but if we choose a different diffeomorphism  of $M \setminus M_0$  to $R^{2d}$ then we obtain a quantization with another reproducing kernel with star product on symbols which satisfy the correpondence principle.  These two quantizations need not be equivalent in the sense that there maynot be  a natural map between the Hilbert spaces which preserve the reproducing  kernel.

The set of meausre zero which we remove is the lower dimensional skeleton in cell decomposition so that what remains is a top dimensional cell which we identify with $U_0 \subset CP^n$, one of the homogeneous charts. 
We pull back the polynomials on $U_0$ to $X$ for the quantization. However we have to keep track of  the cell decomposition because of global holonomy.  The loops may pass through the sets of measure zero in $M_0$ which we have removed.  But this can be handled if we remember the lower dimensional skeleton we had removed.  Thus, even though we remove a set of measure zero,  it plays an important role in detemining the global holonomy.  We illustrate with the torus.

In  this context we recall that in \cite{DeGh} we had considered totally real submanifolds of $CP^n$ and defined pull back operators and their $CP^n$-symbols and showed that they satisfied the correspondence principle.

In this article we also exhibit Berezin-Toeplitz quantization on a compact complex manifold.

This work is part of Kohinoor Ghosh's thesis \cite{DeGh2}.

It  has many interesting applications in harmonic analysis and non-commutative geometry.  This is work in progress.

Role of $CP^n$ or $C^n$ can  taken by  other appropriate manifolds too.

\section{Review of Berezin quantization on $CP^n $}

This section is a review based on ideas from \cite{Be}.
In Berezin \cite{Be},  the quantization on $CP^n$ is achieved thinking of it as a homogeneous space.  In this section we give an explicit  path  to the quantization using a local description. 

Let $\Phi_{FS}$ be a local K$\ddot{\rm{a}}$hler potential for the Fubini-Study  K$\ddot{\rm{a}}$hler form $\Omega_{FS}$ on $CP^n$. Let us recall how this looks in local coordinates.

Let  $U_0 \subset CP^n$ given by $U_0 = \{\mu_0 \neq 0\}$ where $[\mu_0,...., \mu_n]$ are homogeneous  coordinates on $CP^n$.   

Let $\Phi_{FS}(\mu, \bar{\mu}) =  \ln \left( 1 + \sum_{i=1}^n | \mu|^2 \right)$ be the K$\ddot{\rm{a}}$hler potential and the K$\ddot{\rm{a}}$hler metric $G$ is given by  
$g_{ij}^{FS} = \frac{\partial^2 \Phi_{FS}}{\partial \mu_i  \partial \bar{\mu}_j}.$
 
The Fubini-Study form is given by $\Omega_{FS} = \sum_{i, j=1}^n \Omega^{FS}_{ij} d \mu_i \wedge d \bar{\mu}_j, $ where  the K$\ddot{\rm{a}}$hler metric $G$ and the K$\ddot{\rm{a}}$hler form $\Omega_{FS}$ are related by
$\Omega_{FS} (X, Y) = G(I X, Y)$. 

The coefficients of the inverse matrix $\Omega_{FS}^{ij}$ appears in the definition of the Poisson bracket of two functions $t$ and $s$:

$\{t,s\}_{FS} = \sum_{i,j=1}^n \Omega_{FS}^{ij} \left(\frac{\partial t}{\partial \bar{\mu}_j}  \frac{\partial s}{\partial \mu_i}  - \frac{\partial s}{\partial \bar{\mu}_i}  \frac{\partial t}{\partial \mu_j} \right).$

Let $T = \{ (\mu, \nu) \in C^n \times C^n| \mu \cdot \bar{\nu} =-1 \} $ and $S = (C^n \times C^n )\setminus T $.  Note that the diagonal $\Delta \subset S$.   
 For $(\mu, \nu) \in S $, we can define (taking a branch of  the logarithm)
 $\Phi_{FS} (\mu, \bar{\nu})  = \ln \left(1 + \mu \cdot  \bar{\nu}  \right).$

Let $H^{\otimes m}$ be the $m$-th tensor product of the hyperplane bundle $H$ on $CP^n$.  Then  recall that $m \Omega_{FS}$ is its curvature form and $m \Phi_{FS}$ is a local K$\ddot{\rm{a}}$hler potential.
Let $\Gamma_{hol}$ be holomorphic sections on it.  Let  $\{\psi_i \}_{i=1}^N$ be an orthonormal basis for it.  
On $U_0$ the sections of $H^{\otimes m}$ are functions since the bundles are trivial when restricted to $U_0$.

Let $\hbar = \frac{1}{m}$ be a  parameter. 

Thus $\{\psi_i\}$ implicitly depend on $\hbar$. We define

$dV(\mu) =  \left|\Omega_{FS}^n(\mu)|_{U_0}\right| =  {\mathcal G}(\mu)  \Pi_{i=1}^{n} |d \mu_i \wedge d \bar{\mu}_i| = {\mathcal G}(\mu)  |d \mu \wedge d \bar{\mu}| = \frac{|d \mu \wedge d \bar{\mu}|}{(1 + |\mu|^2)^{n+1}} $ to be a volume form on $C^n$, where ${\mathcal G} = \det[g^{ij}|_{U_0}]$.  
 
Then $ V=\int_{C^n} dV = \int_{C^n} \frac{|d \mu \wedge d \bar{\mu}|}{(1 + |\mu|^2)^{n+1} } < \infty$.

Let $(\mu_1, \mu_2,...\mu_n)$ be coordinates on $U_0 \equiv C^n$ such that  $[1, \mu_1, \mu_2, ...,\mu_n]  \in U_0$.

Let  $(c(m))^{-1}  = \int_{U_0} \frac{1} { (1+  |\nu|^2)^m} dV(\nu)= \int_{U_0} e^{-m \Phi_{FS} (\nu,\bar{\nu})} dV(\nu)$ where  recall $e^{m \Phi_{FS} (\nu,\bar{\nu})} = (1 + |\nu|^2)^m$. Also,   
$D_{(q_1, q_2,...q_n;q)} = c(m) \int_{U_0} \frac{|\nu_1|^{2q_1}...|\nu_n|^{2q_n}} {  (1+  |\nu|^2)^m} dV(\nu) ,$ where $q_i's$ are all possible positive  integers such that   $q_1+ ...+q_n = q;q=0,...,m.$ 

Let $\Psi_{(q_1,q_2,...,q_n;q)}(\mu) = \frac{1}{\sqrt{D_{(q_1,...,q_n;q)}}} \mu_1^{q_1}...\mu_n^{q_n}$ where $q_1+ ...+q_n = q;q=0,...,m.$ 

For shorthand we will use $I$ for the multi-index $I_q = (q_1,...,q_n;q)$ which runs over the set 
  $q_1 + ...+ q_n = q; q=0,...,m$.

Then $D_I =   c(m)  \int_{U_0} \frac{|\nu|^{2I}}{ (1+  |\nu|^{2})^m} dV(\nu).$
  
  Let an innerproduct   on the space of functions on $U_0$ be defined as 
  
  $\left<f,g \right> = c(m) \int_{U_0} \frac{ \overline{f(\nu)}g(\nu)}{ (1 + |\nu|^2)^m}  dV(\nu) =  c(m) \int_{U_0} \overline{f(\nu)}g(\nu) e^{-m \Phi_{FS}(\nu,\bar{\nu})}   dV(\nu).$ 
  
It is easy to check that 
 $\{ \Psi_{(q_1,...,q_n; q)} \}$ are   orthonormal in $C^n$ with respect to the inner product defined as above  and are restriction of  a basis for sections of $H^{\otimes m}$  to $U_0$. This forms a Hilbert space.

Let $N$ is the dimension of the Hilbert space,  i.e $N = \sum_J(1) $ where $J$ runs over the indices 
$J = (p_1,...,p_n, p), p_1+...+p_n = p, p=0,...,m$ and  $V = \int_{U_0} |\Omega^n| $.

{\bf Definition:}
The Rawnsley-type coherent states \cite{Ra}, \cite{Sp}  are given on $U_0$  by $\psi_{\mu}$ reading as follows:

$\psi_{\mu} (\nu) := \sum_{{q_1+q_2+...+q_n=q;q=0,1,...,m}} \overline{\Psi_{(q_1,q_2,...,q_n;q)}(\mu)}\Psi_{(q_1,q_2,...,q_n;q)}(\nu).$

In short hand notation 
$\psi_{\mu}  :=   \sum_{{I}} \overline{\Psi_{I}(\mu)}\Psi_{I}.$
\begin{proposition}\label{kernel}
Reproducing kernel property.
If  $\Psi$ is any other section,  then $\left< \psi_{\mu}, \Psi \right> =   \Psi(\mu)$.  In particular, $\left< \psi_{\mu}, \psi_{\nu} \right> =   \psi_{\nu}(\mu)$.
\end{proposition}

\begin{proof}
By linearity,  it is enough to  check this for $\Psi = \Psi_{I_0}$ a basis element.  

$  \left< \psi_{\mu}, \Psi_{I_0} \right> = \left<  \sum_{{I}} \overline{\Psi_{I}(\mu)}\Psi_{I}, \Psi_{I_0} \right> =  \sum_{{I}} \Psi_{I}(\mu) \left< \Psi_{I}, \Psi_{I_0} \right>.$
Now we observe that $ \left< \Psi_{I}, \Psi_{I_0} \right> = \delta_{II_0}$.
Thus $\left< \psi_{\mu}, \Psi_{I_0} \right> =    \Psi_{I_0}(\mu)$.
\end{proof}

\begin{proposition}\label{resolution}
Resolution of identity property:
$$c(m) \int_{U_0} \left< \Psi_1 , \psi_{\mu} \right> \left< \psi_{\mu} , \Psi_2 \right>  e^{-m \Phi_{FS}(\mu,\bar{\mu})} dV(\mu) = \left< \Psi_1 ,  \Psi_2 \right>.$$ In particular, 

$$c(m) \int_{U_0} \left< \psi_{\nu},   \psi_{\mu} \right>\left< \psi_{\mu},  \psi_{\nu} \right>  e^{-m \Phi_{FS}(\mu,\bar{\mu})} dV (\mu)= \left< \psi_{\nu} ,  \psi_{\nu}  \right>. $$
\end{proposition}

\begin{proof}

We know that

$c(m) \int_{U_0} \left< \Psi_1 , \psi_{\mu} \right> \left< \psi_{\mu} , \Psi_2 \right>  e^{-m \Phi_{FS}(\mu,\bar{\mu})} dV =c(m) \int_{U_0} \overline{\Psi_1(\mu)} \Psi_2 (\mu)    e^{-m \Phi_{FS}(\mu,\bar{\mu})} dV  $ since by reproducing kernel property, $\left< \psi_{\mu}, \Psi \right> =    \Psi(\mu)$.  The above integral is $\left< \Psi_1, \Psi_2 \right>$. 
\end{proof}

\medskip

{\bf Notation:} We denote by 
 ${\mathcal L}_{m} (\mu,  \bar{\mu}) = \left< \psi_{\mu}, \psi_{\mu} \right> = \psi_{\mu}(\mu),$ 
 ${\mathcal L}_{m} ( \mu,  \bar{\nu}) =\left< \psi_{\mu}, \psi_{\nu} \right> = \psi_{\nu}(\mu).$

Let $\hat{A}$ be a bounded linear operator acting on ${\mathcal H}$.  Then,  as in  \cite{Be},  one can define a symbol of the operator as
$$A(\nu, \bar{\mu}) = \frac{\left< \psi_{\nu} , \hat{A} \psi_{\mu}, \right>}{ \left<  \psi_{\nu},  \psi_{\mu} \right>}.$$

One can show that one can recover the operator from the symbol by the formula ~\cite{Be}:
$$(\hat{A} f) (\mu) = c(m) \int_{U_0} A(\mu,  \bar{\nu}) f(\nu)  {\mathcal L}_m(\mu,  \bar{\nu})  e^{-m\Phi(\nu, \bar{\nu})} d V(\nu).$$

Let $\hat{A}_1,  \hat{A}_2$ be two such operators and let $\hat{A_1} \circ \hat{A_2}$ be their composition.

Then the symbol of $\hat{A_1} \circ \hat{A_2}$ will be given by the star product defined as in \cite{Be}:
\begin{eqnarray*}\label{star}
& & (A_1 * A_2 ) (\mu, \bar{\mu}) \\
&=&  c(m) \int_{U_0} A_1(\mu,  \bar{\nu}) A_2 (\nu,  \bar{\mu}) \frac{{\mathcal L}_{m} (\mu, \bar{\nu}) {\mathcal L}_{m} (\nu, \bar{\mu})}{ {\mathcal L}_{m} (\mu, \bar{\mu}){\mathcal L}_{m} (\nu, \bar{\nu})}  {\mathcal L}_m(\nu, \bar{\nu}) e^{-m\tilde{\Phi}(\nu, \bar{\nu})} d V(\nu),
\end{eqnarray*}
where recall $\frac{1}{c(m)} = \int_{U_0}  e^{-m \Phi_{FS} (\nu, \bar{\nu})} dV(\nu). $

This is the symbol of $\hat{A_1} \circ \hat{A_2}$.

\medskip

One can show   the following (\cite{DeGh2})
\begin{proposition}\label{multi}
$\psi_{\mu}(\nu) =  ( 1 + \bar{\mu} \cdot \nu)^m$.
\end{proposition}

\begin{equation}\label{A}
 {\mathcal L}_{m} (\mu, \bar{ \mu}) = \left< \psi_{\mu}, \psi_{\mu} \right> = \psi_{\mu}(\mu)  =  e^{m \Phi_{FS}(\mu, \bar{ \mu})}
\end{equation}
 and 
for $(\mu, \nu) \in S$,
\begin{equation}\label{AA}
{\mathcal L}_{m} ( \mu,  \bar{\nu}) =\left< \psi_{\mu}, \psi_{\nu} \right> = \psi_{\nu}(\mu)  =   ( 1 + \mu \cdot \bar{\nu})^m =  e^{m \Phi_{FS}(\mu,  \bar{\nu})}.
\end{equation}

\medskip

Let $(\mu , \nu) \in S$. 

 Then we can define  
 $  \phi_{FS}(\mu, \bar{\mu}| \nu, \bar{\nu}) =  \Phi_{FS} (\mu, \bar{\nu}) +  \Phi_{FS} ( \nu,  \bar{\mu}) - \Phi_{FS} ( \mu, \bar{\mu}) - \Phi_{FS} ( \nu,  \bar{\nu}). $

In fact,
$\phi_{FS}(\mu, \bar{\mu}| \nu, \bar{\nu}) = \ln \big( \frac{(1 + \nu \cdot \bar{\mu}) ( 1 + \mu \cdot \bar{\nu})}{ ( 1 + |\mu|^2) ( 1+ | \nu|^2)} \big).$

\medskip

It is easy to show (\cite{DeGh2})
\begin{proposition} We have $\phi_{FS}$ is non-positive  on $S$  and  has a zero and a non-degenerate critical point ( as a function of $\nu$) at $\nu = \mu$. 
\end{proposition}
\begin{proof}
We know
$\phi_{FS}(\mu, \bar{\mu}| \nu, \bar{\nu}) = \ln \big( \frac{(1 + \nu \cdot \bar{\mu}) ( 1 + \mu \cdot \bar{\nu})}{ ( 1 + |\mu|^2) ( 1+ | \nu|^2)} \big)$.
By Cauchy-Schwartz inequality,  we have $\phi_{FS}$ is non-positive  definite.
 Fixing $\mu$, a straightforward calculation shows that $\frac{\partial \phi_{FS}(\mu, \bar{\mu}| \nu, \bar{\nu})}{\partial \nu_{i}}|_{\nu = \mu} = 0 $ for each $i$ and   $\frac{\partial^2 \phi_{FS}(\mu, \bar{\mu}| \nu, \bar{\nu})}{\partial \nu_{i} \partial \bar{\nu}_{j}}|_{\nu = \mu} \neq 0$. 
\end{proof}

Then we have 
$ \frac{{\mathcal L}_{m} (\mu, \bar{\nu}) {\mathcal L}_{m} (\nu, \bar{ \mu})}{ {\mathcal L}_{m} (\mu, \bar{\mu}){\mathcal L}_{m} (\nu, \bar{\nu})} =   e^{m\phi_{FS}(\mu, \bar{\mu}| \nu, \bar{\nu})},$
where we have used equations  (\ref{A}, \ref{AA}).  
We also have, by the reproducing kernel property,  that 
$\int_{U_0} \frac{{\mathcal L}_{m} (\mu,  \bar{\nu}) {\mathcal L}_{m} (\nu,  \bar{\mu})}{ {\mathcal L}_{m} (\mu, \bar{\mu}){\mathcal L}_{m} (\nu, \bar{ \nu})} d V(\nu) = \frac{1}{c(m)} .$

\medskip

{\bf Theorem:} [Berezin]\label{corresp}

Let $\mu \in  C^n$.  

The star product satisfies the correspondence principle:

1. $ \lim_{m \rightarrow \infty} (A_1 \star A_2)(\mu, \bar{\mu}) = A_1(\mu, \bar{\mu}) A_2 (\mu, \bar{\mu}),$

2. $ \lim_{m \rightarrow \infty} m (A_1 \star A_2 - A_2 \star A_1)(\mu, \bar{\mu}) =  i \{ A_1, A_2\}_{FS} (\mu, \bar{\mu}).$

\medskip
See \cite{Be}, \cite{DeGh2} for proof.

\section{ \bf Berezin-type quantization on compact  even dimensional manifolds}

Let $M^{2d}$ be an even dimensional compact smooth manifold.  We  do not  consider any  symplectic structure or Poisson structure or  group action on it. To obtain a  Berezin-type  quantization on it,  first we embed the manifold  (after perhaps removing a subset of measure zero) in $ CP^d$  and  then induce a local Poisson structure   on the embedded submanifold and induce the Berezin quantization from $CP^d$.  The Hilbert space of quantization is expected to be of finite dimension (since $M$ is compact) and  for that we choose $CP^d$ and not $C^d$.

Let $M^{2d}$ be a compact topological  manifold. Then by \cite{DH},  there exists a skeleton  $M_0$ of dimension at most $2d-1$ such that  $X = M \setminus M_0$  is homeomorphic to $R^{2d}$.  
We assume  $M^{2d}$ is equipped with a differentiable structure such that $M\setminus M_0$ is diffeomorphic to $R^{2d}$ with standard smooth structure. 

Let $\tau$ be the diffeomorphism and  $Y = \tau(X) = R^{2d} \equiv C^d$.  By abuse of notation,  we name the coordinates on $Y$ as $(\tau_1, \tau_2, ...., \tau_d)$ where $\tau_j = x_j + i y_j$, $j=1,...,d$, where $(x_1, y_1,...,x_d, y_d) \in R^{2d}$.  Let  $Y$  be given by the coordinates   $( \tau_1,.... , \tau_d )$.  Let $U_0$ be the open subset of $CP^d$ given by $\{w_0 \neq 0 \}$ where 
$[w_0,..., w_d]$ is a local coordinate on $CP^d$.  Let $U_0 = \{[1,\tau_1,...,\tau_d]\} \equiv C^d $   where $\tau_i = \frac{w_i}{w_0}$,  $i=1,..,d$.

Let us give a metric on $X= M \setminus M_0$ by identifying it with its image $Y= \tau(X) \equiv U_0$.  The volume form is  $dV =  \frac{ | d \tau \wedge d \bar{\tau}| }{(1 + |\tau|^2)^{d+1}}$ and $V = \int_{Y} dV < \infty$.

{\bf Algebra of  operators on $M \setminus M_0$}:

On $M \setminus M_0$, we define the Hilbert space of quantization to be $ \tilde{{\mathcal H}}_{\tau} = \tau^*({\mathcal H_Y})$ (i.e.  pulled back by the diffeomorphism $\tau$), where the volume form on $M \setminus M_0$ is induced from $U_0 \subset CP^d$.   Let $\zeta \in X$. Let $\tau = \tau(\zeta)$,  $s \in {\mathcal H}_Y$.  Let $d S(\zeta) $ be the volume form of $M$ such that $M_0$ is of measure zero.

Let Let $h(\zeta) >0$ be such that $h(\zeta)  dS(\zeta) = d V_{Y}(\tau) = \frac{ | d \tau \wedge d \bar{\tau}| }{(1 + |\tau|^2)^{d+1}}$. In other words, 
$\int_{X} |\tau^*(s)|^2(\zeta) h(\zeta) dS(\zeta) = \int_Y |s|^2 \frac{ | d \tau \wedge d \bar{\tau}| }{(1 + |\tau|^2)^{d+1}}$.

Let $\tilde{s} \in \tilde{{\mathcal H}}_{\tau}$ such that $\tilde{s} = \tau^*(s)$.  Then we define bounded linear operators $\hat{\tilde{A}} $ on ${\mathcal H}_{\tau} $ to be 
$$\hat{\tilde{A}} (\tilde{s})(p) \equiv  \hat{A}(s)(z),$$
where $z = \tau(p) \in U_0$ and $\hat{A}$ is a bounded linear operator on ${\mathcal H_Y}$.  It can be shown  that given a  bounded linear operator $\hat{\tilde{A}} $ on ${\mathcal H}_{\tau} $, there is a unique bounded linear operator $\hat{A}$ on ${\mathcal H_Y}$ such that $\hat{\tilde{A}} (\tilde{s})(p) \equiv  \hat{A}(s)(z)$.

Then symbols and star product can be defined for  $\hat{\tilde{A}} $ via $\hat{A}$ and correspondence principle follows. Now we elaborate this. 

The symbol of    $\hat{\tilde{A}} $ is defined to be 
 $\tilde{A}(p, q)   \equiv A(z, \bar{w}) $ where $z = \tau(p), w = \tau(q)$.
 
 Suppose we have two operators 
$\hat{\tilde{A}}_1 $ and $\hat{\tilde{A}}_2$.

Then $\tilde{A}_1 * \tilde{A}_2   $ is defined on $(M \setminus M_0) \times (M \setminus M_0)$ to be 
$(\tilde{A}_1 * \tilde{A}_2)(p,p) \equiv (A_1 * A_2 )  (z, \bar{z})  $ in $CP^d$.

In general the algebra of operators will depend on the diffeomorphism.

Then we can see that the star product satisfy the correspondence principle. The proof is exactly same as  the previous section with $n=d$.

\begin{proposition}
Let  $\tau \in  C^d$.

$\tilde{A}_1 * \tilde{A}_2 = A_1 * A_2$ satisfy the correspondence principle.

1. $ \lim_{m \rightarrow \infty} (A_1 \star A_2)(\tau, \bar{\tau}) = A_1(\tau, \bar{\tau}) A_2 (\tau, \bar{\tau}),$

2. $ \lim_{m \rightarrow \infty} m (A_1 \star A_2 - A_2 \star A_1)(\tau, \bar{\tau}) =  i \{ A_1, A_2\}_{FS} (\tau, \bar{\tau}).$

\end{proposition}

\begin{proof}
Set  $n=d$ in the previous section. The proof follows  essentially from  Lemma (2.1)  in ~\cite{Be} as elaborated in \cite{DeGh2}. 
\end{proof}

\subsection{ Equivalence of two Berezin quantizations:}

 On a smooth (complex) manifold $M^{2d} \setminus M_0 $, let there be a local Poisson structure and a  Berezin-type quantization defined as above  induced from $CP^d$.  Suppose there are two diffeomorphisms (biholomorphisms if $M \setminus M_0$ is complex)  which induce two such quantizations.  
   Then there are    two Hilbert spaces with reproducing kernels and star products on symbols of bounded linear operators  which satisfy the correspondence principle.  Suppose there exists  a  smooth (or biholomorphic)  bijective map  $\psi$ from $M \setminus M_0 $ to $M \setminus M_0$ which preserve  the  local Poisson structures.  If $\psi$  induces an isomorphism (i.e.  a bijective linear map that preserves innerproduct) between the two Hilbert spaces such that  the reproducing kernel maps to the corresponding reproducing kernel then we shall say the two Berezin quantizations are equivalent.

 \section{{\bf Our method of quantization for the torus}}

Let $L$ be a line bundle on $ CP^1$.  $CP^1 $ is homeomorphic to the sphere of radius $1$ and let $N, S$ be the north and the south poles and $E$ the equatorial circle.   Let $U_N = S^2 \setminus N$ and $U_S = S^2 \setminus S$ be two charts on the sphere such that $U_N $ is homeomorphic to the equatorial plane using the stereographic projection from $N$ and $U_S$ being the same using stereographic projection from $S$.   The transition function $t_{NS}$ of the line bundle $L$ when restricted to the equator,  winds  the equatorial circle $E$  to $r$ times $U(1)\equiv S^1$ , $r \in Z$.  This winding number characterises  smooth line bundles on the sphere.  For the transition function of  $H$, the hyperplane line bundle,  let the winding number be $r_0$.   Then for  $L= H^{\otimes m}$ the winding number is   $q=r_0 m$.  (As an aside,  the set of holomorphic sections of $H^{\otimes m}$ are in one to one  correspondence with polynomials of degree $\leq m$ in one complex variable--for more details,  see \cite{Nai}, p 500). 
  
  Let $i \theta_1$ be  the imaginary valued connection $1$-form for $H$ (curvature proportional to the Fubini-Study form $\omega_{FS}$).  Let $ i\theta = m i \theta_1$  the  connection $1$-form on $H^{\otimes m}$. Let $m=2s$ be an even integer.  
  Let $\psi$ be a section of  $H^{\otimes m}$ on sphere which satisfies $(d + m i   \theta_1) \psi =0$.   On integration on any closed loop $C_1$ on the sphere,  $\psi = \exp(-im \int_{C_1} \theta_1)  \psi_0$ where $\psi_0= \psi(t_0)$.   The phase factor is called holonomy and is well defined along this path  because the curvature of the line bundle  $\omega_{FS} = d \theta_1$ belongs to the  integral cohomolgy $H^2(S^2, Z)$,  in  \cite{W}, p 158.

  Let $\overline{U_u}$ be the upper hemisphere of the sphere with boundary $E$. $U_u$ is the interior of $\overline{U_u} $ which is diffeomorphic to a disc. 
  As before let  $\psi$ be a section of $H^{\otimes m}$ and $E$ be parametrized by $t$ such that $E_1$ and $E_2$ are parametrised by $0\leq t \leq \frac{1}{2}$ and $\frac{1}{2} \leq t \leq 1$ respectively such that $E = E_1+ E_2$.   Let $\bar{E}_2 = -E_2$,  i.e.  $E_2$ with the reverse direction.
  
  We note that $\exp(-i \int_E \theta) = \exp(-i \int_E 2s \theta_1) = \exp(-i \int_{U_u} 2 s \omega_{FS}) = \exp(-i s \frac{A_{S^2}}{2}) = 1$  since $\int_{U_u} \omega_{FS} = \frac{A_{S^2}}{4} =\pi$. Also
  $\exp(-i \int_{E_1} \theta)  \exp(i \int_{\bar{E}_2} \theta) = 1$ by the same reason. 
Thus $\exp(-i \int_{E_1} \theta)  = \exp(-i \int_{\bar{E}_2} \theta)$.

One sees that $\psi(\frac{1}{2}) = \exp({-i \int_{E_1} \theta}) \psi_0   =    exp({-i 2s \int_{E_1} \theta_1}) \psi_0 $ and $\psi(1)= \exp({-i \int_{E_2} \theta}) \psi(\frac{1}{2})   =     \exp({-i \int_{E_2} \theta})    \exp({-i \int_{E_1} \theta}) \psi_0 =      \exp({-i \int_{E} \theta})  \psi_0=  \exp({-i \int_{E}2  s \theta_1}) \psi_0 = \psi_0  $.

Let $E = E_1^1 + E_1^2 $ where $E_1^1 $ is half way of $E_1$ and $E_1^2$ is the other half of $E_1$ and $E_2 = E_2^1 + E_2^2$, where $E_2^1$ is half way of $E_2$ and $E_2^2$ is the other half of $E_2$. 
  
Let $A$ and $B$ be the two representatives of the homology of the torus, $M_0 = A \cup B$.  Let $X = T^{2} \setminus M_0 \equiv U_u$ (by a diffeomorphism).   $\overline{U_u} \setminus U_u$ is $E$ and $T^2 \setminus X = A \cup B$. 
Let us identify points on $E$ with $A \cup B $ such that one-fourths  of the equator $E_1^1$ is identified with $A$  and the other half $E_1^2$ is identified with $B$ (using a quotient map). Similarly, $E_2 = E_2^1 + E_2^2$, such that $E_2^1$ and $E_2^2$   are identifed with $-A$, $-B$. This is in keeping with the cell decomposition of the torus.  $E_1$ and $E_2$ are identified with loops $A+B$ and $-A-B$.

After identification,   $\exp({-i \int_{A+B} \theta}) =  \exp({-i \int_{E_1} \theta}) $ and $\exp({-i \int_{-A-B} \theta}) =  \exp({-i \int_{E_2} \theta})  $.  In other words,  $\exp({-i \int_{E_1 + E_2} \theta}) =  \exp({-i \int_{E}2  s \theta_1})=1$, as was shown before.   

Take a loop $C$ on the torus  such that $C = k_1 A + k_2 B$.      
This is identified with $\tilde{C} = k_1 E_1^1 + k_2 E_1^2$.

$C = k_1 A+ k_2 B$ be a closed loop of the torus,  as before, parametrised by $0\leq t \leq 1$ and $\psi$ be a section of $H^{\otimes m}$, $m=2s$, an even integer.  Then $\psi(1)  = \exp({-i \int_{k_1 A}  \theta - i \int_{k_2 B}  \theta}  ) \psi_0$ where $\psi_0= \psi(0)$.  Then the phase factor due to holonomy  is $\exp({-i  \int_{k_1 E^1_1} \theta -i   \int_{k_2 E_1^2} \theta} ) $
which is well defined (because we can translate the question to that on the sphere).

If $p \in X= T^2 \setminus (A \cup B)$,  $\beta$ is a loop on the torus $T^2$ which is contained entirely in $T^2 \setminus (A \cup B)$ one can easily show there is well defined global holonomy,  after the identification of $X$ with $U_u$.   If $\beta$ is a loop starting and ending at $p \in T^2$ that intersects $A$ or $B$,  by our identification of $A$ and $B$ with $E_1^1$ and $E_1^2$,  the holonomy on $\beta$ is also well defined (as we can translate the question to that on the sphere).

\section{ {\bf Toeplitz quantization on  compact complex manifolds } }

Let $M$ be a compact complex manifold of dimension $d$. Let $M_0$ be a set of measure zero such that $ X= M \setminus M_0 \equiv R^{2d} \equiv C^d$ (diffeomorphism).  Let $U_0=\{[1, z_1, z_2, ..., z_d]\}$ be one of the homogeneous  charts of $CP^d$. 
Thus we have  an embedding,   $\epsilon$,  which maps  $M \setminus M_0$ onto $U_0 \subset CP^d $.  Note that  $CP^d$ is endowed with Fubini-Study metric. 

Recall that the volume element on $CP^d$ restricted to $U_0$ is given by $dV_{CP^d} =dV(\mu)= \frac{|d \mu \wedge d \bar{\mu}|}{(1 + |\mu|^2)^{d+1}}$.

Let $H$ be the hyperplane line bundle on $CP^d$ and $H^{\otimes m}$ be the $m$-th tensor power of $H$ and $\mathcal{H}^m$ be the Hilbert space of square integrable holomorphic sections of $H^{\otimes m}$ restricted on $U_0$. Let $\mathcal{H}_{X}^m$ denotes $\epsilon^\ast (\mathcal{H}^m)$.

 $M \setminus M_0$ has an induced  volume form  as follows.  Let $\Sigma = \epsilon(M \setminus M_0)$ and   $h(\zeta)  dS(\zeta) = d V_{\Sigma}(\epsilon(\zeta))$, where $h >0$ is a smooth function.  Note that  all pullback sections  in
 $\mathcal{H}_{X}^m$  are  square integrable w.r.t.  the measure $h dS$ on $X = M \setminus M_0$.

Let $f, g$ be a smooth function on $CP^d$ restricted to $U_0$ and let $\tilde{f}, \tilde{g}$ be the smooth functions on  $M \setminus M_0$, which are pulled back by $\epsilon$, i.e., for $\mu\in M \setminus M_0$,  $\tilde{f}(\mu) \doteq f(\epsilon(\mu))$, similarly  $\tilde{g}(\mu) \doteq g(\epsilon(\mu))$. 

We claim $f$ is a unique function on $CP^d$ given $\tilde{f} = \epsilon^*(f)$,.  Suppose,  $\tilde{f} = \epsilon^*(f_1) = \epsilon^*(f_2)$. Then $f_1 - f_2=0$ on $ \Sigma = \epsilon(M \setminus M_0)$.  But   $\Sigma = U_0 \subset CP^d$ is an open set in $CP^n$ such that its complement is of measure zero.   Since $f_1 - f_2 $ is smooth,  it extends to all of $CP^d$  and is identically $0$.

Recall for  $CP^d$ (restricted to $U_0$),   $m$-th level Toeplitz operator of $f$, denoted by $T^m_f$, defined on  $\mathcal{H}^m$, defined as $T^m_f (s)= \Pi^m(f s)$, where $\Pi^m$ is the projection map from square integrable smooth sections onto $\mathcal{H}^m$ and $s\in \mathcal{H}^m$. 
Let $\tilde{s}= \epsilon ^\ast s$.     One can show that given $\tilde{s} $, $s$ is unique.  
This is because if $\tilde{s} = \epsilon ^\ast s_1 = \epsilon ^\ast s_2$. Then $s_1 - s_2 =0$ on $\Sigma$ and thus on $CP^d$.   But $s_1, s_2$ are global holomorphic sections of $H^m$ and can be extended to all of $CP^d$. Thus $s_1 - s_2  \equiv 0$.

For $ X = M \setminus M_0$, we denote $$|| \tilde{s} ||^2 = || \tilde{s}||^2_X = \int_{X} |\tilde{s}|^2  h(\zeta)  dS(\zeta) = \int_{\Sigma} | s_{|_ \Sigma}|^2 d V_{\Sigma}(\epsilon(\zeta))$$ where recall $\Sigma = \epsilon(X)$. 

But $ \int_{\Sigma} | s_{|_ \Sigma}|^2 d V_{\Sigma}(\epsilon(\zeta)) = \int_{CP^d} | s|^2 d V_{CP^d}$ since $ \Sigma = U_0 \subset  CP^d $ and $CP^d \setminus U_0$ is of measure zero.

Thus we have 
\begin{equation}\label{norms}
|| \tilde{s} ||^2 = ||s||^2
\end{equation}
 where the first norm is in $X= M \setminus M_0$ and second norm is in $CP^d$.

For a functions $\tilde{f} \in C^{\infty}(X)$, we define a set of  operators for $M \setminus M_0$, defined on $\mathcal{H}_{X}^m$, denoted by $\tilde{T}^m_{\tilde{f}}$.

{\bf Definition}
$\tilde{T}^m_{\tilde{f}} (\tilde{s})= \tilde{\Pi}^m(\tilde{f} \tilde{s})$where 
$\tilde{\Pi}^m  \epsilon^{\ast}\doteq \epsilon^\ast \Pi^m$.

Since $\tilde{f} \tilde{s} = \epsilon^*(f s)$ for a unique $fs \in {\mathcal H}^m$,  we have that $\tilde{T}^m_{\tilde{f}} (\tilde{s})$ is well defined. 

We know from Toeplitz quantization of $CP^d$ (see \cite{BMS}), that, 

\begin{equation}\label{toeplitzcpn}
\lim_{m \rightarrow \infty} || T^m_{f} || = ||f||_{\infty},\\
\lim_{m \rightarrow \infty} || m [ T^m_{f} , T^m_{g}]  - i T^m_{\{f, g\}}|| = 0.
\end{equation}

Let $\{\tilde{f}, \tilde{g}\} \doteq \epsilon^\ast \{f,g\}$.

\begin{proposition}
$\tilde{T}^m_{\tilde{g}} \epsilon^\ast= \epsilon^\ast T^m_g$ and $\tilde{T}^m_{\{\tilde{f},\tilde{g}\}} \epsilon^\ast= \epsilon^\ast T^m_{\{f,g\}}$
\end{proposition}

\begin{proof}

To prove the first equality, 
\begin{eqnarray*}
 \tilde{T}^m_{\tilde{g}} \epsilon^\ast s (\mu) &=& \tilde{\Pi}^m(\tilde{g}\cdot \epsilon^\ast (s))(\mu)= \tilde{\Pi}^m(\epsilon^\ast g\cdot \epsilon^\ast(s))(\mu)= \tilde{\Pi}^m(\epsilon^\ast (g\cdot s) ) (\mu)\\
 &=& (\epsilon^{\ast} \Pi^m g \cdot s) (\mu) \\
 &=&(\epsilon^\ast T^m_g)(s)(\mu)
\end{eqnarray*}
The second equality follows from this.  
\end{proof}
	
\medskip

Recall that $||\epsilon^*(s)|| = ||\epsilon^*(s)||_X$ and $||s|| = || s||_{CP^d}$. 
Now $||\epsilon^*(s)|| = ||s||$ by (\ref{norms}). This implies $||\tilde{T}^m_{\tilde{f}} || = || T_{f}^m||$ for each $m$.
\begin{equation} 
\lim_{m \rightarrow \infty} || \tilde{T}^m_{\tilde{f}} || =  \lim_{m \rightarrow \infty} || T^m_{f} ||.
 \end{equation}
\begin{proposition}
$[\tilde{T}^m_{\tilde{g}}, \tilde{T}^m_{\tilde{g}}]= \epsilon^*[ T_f^m, T^m_g].$
\end{proposition}

\begin{proof}
Recall $ \tilde{\Pi}^m  \epsilon^{\ast}\doteq \epsilon^\ast \Pi^m   $. 
$\tilde{\Pi}^m(\tilde{f}\tilde{\Pi}^m \tilde{g}\tilde{s}) = \tilde{\Pi}^m(\tilde{S_1})$ where $S_1 = f \cdot \Pi^m g \cdot s$ and $   \tilde{S_1} = \tilde{f}\tilde{\Pi}^m \tilde{g}\tilde{s} = \epsilon^* (S_1)$.

Then, since $ \tilde{\Pi}^m  \epsilon^{\ast}\doteq \epsilon^\ast \Pi^m   $,  we have 
$ \tilde{\Pi}^m (\tilde{S_1}(\mu)) = \epsilon^{\ast} \Pi^m ( S_1 )  = \epsilon^{\ast} \Pi^m (   f \cdot \Pi^m g \cdot s     ) $

Interchanging $f$ and $g$ we have 
$ \tilde{\Pi}^m(\tilde{g}\tilde{\Pi}^m \tilde{f}\tilde{s}) = \epsilon^{\ast} \Pi^m (   g \cdot \Pi^m f \cdot s     )$

Thus $\tilde{\Pi}^m(\tilde{f}\tilde{\Pi}^m \tilde{g}\tilde{s}) - \tilde{\Pi}^m(\tilde{g}\tilde{\Pi}^m \tilde{f}\tilde{s} ) = \epsilon^{\ast} \Pi^m (   f \cdot \Pi^m g \cdot s     ) - \epsilon^{\ast} \Pi^m (   g \cdot \Pi^m f \cdot s     ).$ 
\end{proof}

\begin{proposition}\label{toeplitzM}

$\lim_{m \rightarrow \infty} || \tilde{T}^m_ {\tilde{f} }|| = \lim_{m \rightarrow \infty} || T^m_f || = || f||_{\infty}. $

 $\lim_{m \rightarrow \infty} || m [ \tilde{T}^m_{\tilde{f}} , \tilde{T}^m_ {\tilde{g}}]  - i \tilde{T}^m_{\{\tilde{f}, \tilde{g}\}}|| = 0 .$
\end{proposition}

\begin{proof}
As seen before $||\epsilon^*(s)|| = ||\epsilon^*(s)||_X = || s||_{CP^d} =||s||$. This implies $||\tilde{T}_{\tilde{f}} || = || T_{f}||$.

Rest follows from the previous two propositions.  This ends the proof.
\end{proof}

\medskip


If we have two biholomorphisms $\epsilon_1$ and  $\epsilon_2$ from  $M \setminus M_0$  to $U_0 \subset CP^d$  we have an equivalent Toeplitz quantization because we can define an equivalence of the Hilbert spaces  and the Poisson bracket is also preserved. ~\cite{DeGh2}.

\section{\bf A conjecture}

Let $M$ be a compact integral K$\ddot{\rm{a}}$hler manifold with  K$\ddot{\rm{a}}$hler form $\omega$.  Let $L$ be a geometric quantum bundle whose curvature is proportional to $\omega$.  Let us consider  $L^{\otimes m}$ 
whose curvature is proportional to  $m \omega$.

There is a corollary to a theorem by several mathematicians \cite{Ti, Ru, Ze, Lu, Ca}  as mentioned in \cite{Ga}, page 131. It goes as follows.

{\bf Theorem}[Tian,  Ruan, Zelditch,  Lu, Catlin]  \label{TRZLC}
For large $m,$  an orthonormal basis of $H^0(M, L^m)$ gives a map $\epsilon_m: M \mapsto CP^{N_m},$ where $N_m + 1 = \rm{dim} H^0 (M,L^m)$  and

$ \frac{1}{m} \epsilon_m^*( \Omega_{FS})  - 2 \pi \omega  = O(m^{- 2} )$ in $C^{\infty}$.

\medskip

\medskip 

Let $\Omega_m = \epsilon_m^*(\Omega_{FS})$.

For each $m$ and $t,s$ two smooth functions on $M$,  let the two Poisson brackets on $M$ be 

$\{t, s\}_{PB1} = \sum_{ij} \Omega_{m}^{ij} \left(\frac{\partial t}{\partial \bar{\tau}_i}  \frac{\partial s}{\partial \tau_j}  - \frac{\partial s}{\partial \bar{\tau}_i}  \frac{\partial t}{\partial \tau_j} \right)$ and

$\{t, s\}_{PB2} = \sum_{ij} m \omega^{ij} \left(\frac{\partial t}{\partial \bar{\tau}_i}  \frac{\partial s}{\partial \tau_j}  - \frac{\partial s}{\partial \bar{\tau}_i}  \frac{\partial t}{\partial \tau_j} \right)$

Then motivated  by the theorem above we have the following conjecture.

{\bf Conjecture:}
 Suppose $M$ has a Berezin quantization as defined in \cite{Be} induced by   the K$\ddot{\rm{a}}$hler  form $\omega$.
Also in \cite{DeGh} we defined a pullback Berezin-type  quantization on $M$,  in this case pull back from $CP^{N_m}$.  (One can show that one doesnot need  the totally real condition in this case). We conjecture that these two quantizations are equivalent, in the sense that there is in isomorphism of the Hilbert spaces with reproducing kernels and the Possion brackets (which appear in the correspondence principle) are the  same  in the limit $m \rightarrow \infty$.

\section{Acknowledgement}
The authors would like to thank Professors Purvi Gupta (I.I.Sc),  Pranav Pandit (I.C.T.S.-T.I.F.R.),  Indranil Biswas (T.I.F.R. ),  E. K. Narayanan (I.I.Sc), Rajesh Gopakumar (I.C.T.S.-T.I.F.R.), Mahan Maharaj  (T.I.F.R.) and Shiraz Minwalla (T.I.F.R.) for interesting discussions. 
Rukmini Dey acknowledges support from the project RTI4001, Department of Atomic Energy, Government of India and support from grant
CRG/2018/002835, Science and Engineering Research Board, Government of India.
This work is part of Kohinoor Ghosh's thesis.


\begin{thebibliography}{9}


\bibitem{Be} Berezin F A  1974 Quantization {\it Math USSR Izv.} {\bf 8} 1109









\bibitem{BMS} Bordemann M,  Meinrenkene E and Schlichenmaier M 1994  Toeplitz Quantization of K$\ddot{\rm{a}}$hler manifolds and $gl(N)$, $N \rightarrow \infty$ limits  {\it Comm.  Math. Phys. } {\bf 165} 281

\bibitem{Ca} Catlin D 1997 The Bergman kernel and a theorem of Tian   {\it Analysis and geometry in
several complex variables} (Trends Math. ) (Boston: Birkhauser Boston Inc.)01



\bibitem{DeGh}
 Dey R and Ghosh K 2022 
Pull back coherent states and squeezed states and quantization {\it Symm.  Integr.  Geom.: Meth. and Appl. }{\bf 8} {\bf 028} 01


	
\bibitem{DH}
 Doyle P H and  Hocking J G 1962   A Decomposition Theorem for $n$-dimensional manifolds {\it 
  Proc.  Amer. Math.  Soc.} {\bf 13} 469
 
 \bibitem{E}
Englis M 1996  Berezin Quantization and Reproducing Kernel on Complex Domains {\it  Trans. Amer. Math. Soc. }
 {\bf 348} 411 
 
 
\bibitem{Ga} Gabor S 2014 {\it  An introduction to extremal K$\ddot{\rm{a}}$hler metrics} (Graduate Studies in Mathematics) vol 152 ( Providence: American Mathematical Society)
 
 \bibitem{DeGh2} Ghosh K 2023 Berezin-type quantization of even-dimensional maniflds and pullback coherent states,  {{\it Thesis ICTS-TIFR} https://thesis.icts.res.in/}
	

 
 
 
 
 
 \bibitem{Ko} Kordyukov Y A 2022 Berezin-Toeplitz quantization associated with higher Landau levels of the Bochner Laplacian {\it J. Spectr. Theory} {\bf 12} 143
 
 \bibitem{Lu}  Lu Z 1998 On the lower order terms of the asymptotic expansion of Tian-Yau-Zelditch{\it Amer. J. Math.} {\bf 122} 235

\bibitem{Nai} Nair V P 2004 {\it Quantum Field Theory: A Modern Persective}(New York: Springer)
	

 
 

 \bibitem{Rad}Radcliffe J M 1971  Some Properties of Coherent Spin States {\it J. Phys. A: Gen. Phys.} {\bf 4} 313
	
 \bibitem{Ra} Rawnsley J H 1977 Coherent States and K$\ddot{\rm{a}}$hler Manifolds {\it Quart. J. Math.} {\bf 28}  403	
	
\bibitem{Ru} Ruan W D 1998 Canonical coordinates and Bergman metrics {\it  Comm. Anal. Geom.} {\bf 6} 589


 




 


 
 


	





	
\bibitem{Sp} Spera M 2000
 On K$\ddot{\rm{a}}$hlerian Coherent States {\it Proc.  Int. Conf.  Geom.Integr. Quant} (Sofia: Bulg. Acad.  Scien.) 
 
 \bibitem{Ti}Tian G 1990  On a set of polarized K\"ahler metrics on algebraic manifolds {\it J. Diff. Geom. } {\bf 32} 99
  
 \bibitem{To}Tong D 2016 The Quantum Hall Effect {\it Tata Infosys Lectures} http://www.damtp.cam.ac.uk/user/tong/qhe.html
 


\bibitem{W}Woodhouse N M J  2007 {\it  Geometric Quantization } (Ox. Math.  Mono.) (Oxford : Claredon Press) 
.
 
 \bibitem{Ze} Zelditch S 1998 Szego kernel and a theorem of Tian {\it Int. Math. Res. Notices} {\bf 6} 317


   
   
   










   
   
   

\end{thebibliography}
\end{document}